\documentclass[journal]{IEEEtran}

\usepackage{enumerate}
\usepackage{cite}
\usepackage{amsmath}
\usepackage{amscd}
\usepackage{amssymb}
\usepackage{latexsym}
\usepackage{array}
\usepackage{tabularx}
\usepackage[ruled]{algorithm}
\usepackage[noend]{algorithmic}
\usepackage{epsfig}
\newtheorem{lem}{Lemma}

\newtheorem{deft}{Definition}
\newtheorem{prop}{Proposition}
\newtheorem{rem}{Remark}

\newcommand{\bm}[1]{\mbox{\boldmath{$#1$}}}

\begin{document}
%
\title{Variants of the LLL Algorithm in Digital Communications: Complexity Analysis and Fixed-Complexity Implementation}
\author{Cong~Ling,~\IEEEmembership{Member,~IEEE,}~Wai~Ho~Mow,~\IEEEmembership{Senior Member,~IEEE,}
        and~Nick~Howgrave-Graham
\thanks{This work was presented in part at the IEEE International
Symposium on Information Theory, Nice, France, 2007, the LLL+25 Conference, Caen, France, 2007, the
Information Theory and Applications Workshop, San Diego,
USA, 2008, and the Information Theory Workshop, Taormina, Italy, 2009.}
\thanks{C. Ling is with the Department of Electrical and Electronic Engineering, Imperial College London,
London SW7 2AZ, United Kingdom (e-mail: cling@ieee.org).}
\thanks{W. H. Mow is with the Department of Electronic and Computer Engineering, Hong Kong University of Science and
Technology, Clear Water Bay, Hong Kong, China (e-mail: w.mow@ieee.org).}
\thanks{N. Howgrave-Graham is with the Ab Initio Software Corporation, Lexington, MA 02421, USA (e-mail: nhowgravegraham@abinitio.com).}
}


\maketitle

\begin{abstract}
The Lenstra-Lenstra-Lov\'asz (LLL) algorithm is the most practical lattice reduction algorithm in
digital communications. In this paper, several variants of the LLL algorithm with either lower
theoretic complexity or fixed-complexity implementation are proposed and/or analyzed. Firstly, the
$O(n^4\log n)$ theoretic average complexity of the standard LLL algorithm under the model of i.i.d.
complex normal distribution is derived. Then, the use of effective LLL reduction for lattice
decoding is presented, where size reduction is only performed for pairs of consecutive basis
vectors. Its average complexity is shown to be $O(n^3\log n)$, which is an order lower than
previously thought. To address the issue of variable complexity of standard LLL, two
fixed-complexity approximations of LLL are proposed. One is fixed-complexity effective LLL, while
the other is fixed-complexity LLL with deep insertion, which is closely related to the well known
V-BLAST algorithm. Such fixed-complexity structures are much desirable in hardware implementation
since they allow straightforward constant-throughput implementation.
\end{abstract}



%
\IEEEpeerreviewmaketitle

\section{Introduction}


Lattice pre/decoding for the linear multi-input multi-output (MIMO) channel is a problem of high
relevance in single/multi-antenna, broadcast, cooperative and other multi-terminal communication
systems \cite{mow:IT,viterbo,Damen,agrell}. Maximum-likelihood (ML) decoding for a rectangular
finite subset of a lattice can be realized efficiently by sphere decoding \cite{viterbo,fincke},
whose complexity can nonetheless grow prohibitively with dimension $n$ \cite{jalden}. The decoding
complexity is especially felt in coded systems, where the lattice dimension is larger
\cite{Oggier}. Thus, we often have to resort to approximate solutions, which mostly fall into two
main streams. One of them is to reduce the complexity of sphere decoding, notably by relaxation or
pruning; the former applies lattice reduction pre-processing and searches the infinite
lattice\footnote{It is worth mentioning that, in precoding for MIMO broadcast \cite{Hochwald05} and
differential lattice decoding \cite{ling:jsac}, the lattice is indeed infinite. Infinite lattice
decoding is also necessary when boundary control is difficult.}, while the latter only searches
part of the tree by pruning some branches. Another stream is lattice reduction-aided decoding
\cite{yao,Windpassinger2}, which was first proposed by Babai in \cite{Babai}, which in essence
applies zero-forcing (ZF) or successive interference cancelation (SIC) on a reduced lattice. It is
known that Lenstra, Lenstra and Lov\'asz (LLL) reduction combined with ZF or SIC achieves full
diversity in MIMO fading channels \cite{Taherzadeh:IT,XiaoliMa08} and that lattice-reduction-aided
decoding has constant gap to (infinite) lattice decoding \cite{LingIT07}. It was further shown in
\cite{Jalden:LR-MMSE} that minimum mean square error (MMSE)-based lattice-reduction aided decoding
can achieve the optimal diversity and multiplexing tradeoff. In \cite{wubbenMMSE}, it was shown
that Babai's decoding using MMSE provides near-ML performance for small-size MIMO systems. More
recent research further narrowed down the gap to ML decoding by means of sampling \cite{Shuiyin10}
and embedding \cite{Luzzi10}.

As can be seen, lattice reduction plays a crucial role in MIMO decoding. The celebrated LLL
algorithm \cite{LLL} features polynomial complexity with respect to the dimension for any given
lattice basis but may not be strong enough for some applications. In practice of cryptanalysis
where the dimension of the lattice can be quite high, block Korkin-Zolotarev (KZ) reduction is
popular. Meanwhile, LLL with deep insertions (LLL-deep) is a variant of LLL that extends swapping
in standard LLL to nonconsecutive vectors, thus finding shorter lattice vectors than LLL
\cite{SchnorrEuchner}. Lately, it is found that LLL-deep might be more promising than block-KZ
reduction in high dimensions \cite{GaNg08} since it runs faster than the latter.

Lattices in digital communications are complex-valued in nature. Since the original LLL algorithm
was proposed for real-valued lattices \cite{LLL}, a standard approach to dealing with complex
lattices is to convert them into real lattices. Although the real LLL algorithm is well understood,
this approach doubles the dimension and incurs more computations. There have been several attempts
to extend LLL reduction to complex lattices \cite{Napias96, Howgrave-GrahamThesis, mow03:wcmc,
ComplexLLL}. However, complex LLL reduction is less understood. While our recent work has shown
that the complex LLL algorithm lowers the computational complexity by roughly 50\%
\cite{ComplexLLL}, a rigorous complexity analysis is yet to be developed. In this paper, we analyze
the complexity of complex LLL and propose variants of the LLL algorithm with either lower theoretic
complexity or a fixed-complexity implementation structure.

More precisely, we shall derive the theoretic average complexity $O(n^4\log n)$ of complex LLL,
assuming that the entries of $\mathbf{B}$ are be i.i.d. standard normal, which is the typical MIMO
channel model. For integral bases, it is well known that the LLL algorithm has complexity bound
$O(n^4\log B)$, where $B$ is the maximum length of the column vectors of basis matrix $\mathbf{B}$
\cite{LLL}. The complexity of the LLL algorithm for {\it real or complex-valued bases} is less
known.  To the best of our knowledge, \cite{Daude} was the only work prior to ours
\cite{LingISIT07} on the complexity analysis of the real-valued LLL algorithm for a {\it
probabilistic} model. However, \cite{Daude} assumed basis vectors drawn independently from the unit
ball of $\mathbb{R}^n$, which does not hold in MIMO communications.

Then, we propose the use of a variant of the LLL algorithm---{\it effective LLL reduction} in
lattice decoding. The term {\it effective LLL reduction} was coined in \cite{Howgrave-Graham97},
and was proposed independently by a number of researchers including \cite{Mow:Thesis}. We will show
the average complexity bound $O(n^3\log n)$ in MIMO, i.e., an order lower than that of the standard
LLL algorithm. This is because a weaker version of LLL reduction without full size reduction is
often sufficient for lattice decoding. Besides, it can easily be transformed to a standard
LLL-reduced basis while retaining the $O(n^3\log n)$ bound.

A drawback of the traditional LLL algorithm in digital communications is its variable complexity.
The worst-case complexity could be quite large (see \cite{Jalden08} for a discussion of the
worst-case complexity), and could limit the speed of decoding hardware. To overcome this drawback,
we propose two fixed-complexity approximations of the LLL algorithm, which are based on the
truncation of the parallel structures of the effective LLL and LLL-deep algorithms, respectively.
When implemented in parallel, the proposed fixed-complexity algorithms allow for higher reduction
speed than the sequential LLL algorithm.

In the study of fixed-complexity LLL, we discover an interesting relation between LLL and the
celebrated vertical Bell-labs space-time (V-BLAST) algorithm \cite{vblast:conf}. V-BLAST is a
technique commonly used in digital communications that sorts the columns of a matrix for the
purpose of better detection performance. It is well known that V-BLAST sorting does not improve the
diversity order in multi-input multi-output (MIMO) fading channels; therefore it is not thought to
be powerful enough. In this paper, we will show that V-BLAST and LLL are in fact closely related.
More precisely, we will show that if a basis is both sorted in a sense closely related to V-BLAST
and size-reduced, then it is reduced in the sense of LLL-deep.

{\it Relation to prior work:} The average complexity of real-valued effective LLL in MIMO decoding
was analyzed by the authors in \cite{LingISIT07}. Jald\'en et al. \cite{Jalden08} gave a similar
analysis of complex-valued LLL using a different method, yet the bound in \cite{Jalden08} is less
tight. A fixed-complexity LLL algorithm was given in \cite{VetterLLL}, while the fixed-complexity
LLL-deep was proposed by the authors in \cite{LingITW09}. The fixed-complexity LLL-deep will also
help demystify the excellent performance of the so-called DOLLAR (double-sorted low-complexity
lattice-reduced) detector in \cite{Waters:Dollar}, which consists of a sorted QR decomposition, a
size reduction, and a V-BLAST sorting. Both its strength and weakness will be revealed in this
paper.

The rest of the paper is organized as follows. Section II gives a brief review of LLL and LLL-deep.
In Section III, we present the complexity analysis of complex LLL in MIMO decoding. Effective LLL
and its complexity analysis are given in Section IV. Section V is devoted to fixed-complexity
structures of LLL. Concluding remarks are given in Section VI.

\textit{Notation}: Matrices and column vectors are denoted by upper and lowercase boldface letters
(unless otherwise stated), and the transpose, Hermitian transpose, inverse of a matrix $\mathbf{B}$
by $\mathbf{B}^T$, $\mathbf{B}^H$, $\mathbf{B}^{-1}$, respectively. The inner product in the
complex Euclidean space between vectors $\mathbf{u}$ and $\mathbf{v}$ is defined as $\langle
\mathbf{u}, \mathbf{v}\rangle = \mathbf{u}^H\mathbf{v}$, and the Euclidean length
$\|\mathbf{u}\|=\sqrt{\langle \mathbf{u}, \mathbf{u} \rangle}$. $\Re(x)$ and $\Im(x)$ denote the
real and imaginary part of $x$, respectively. $\lceil x \rfloor$ denotes rounding to the integer
closest to $x$. If $x$ is a complex number, $\lceil x \rfloor$ rounds the real and imaginary parts
separately. The big O notation $f(x)=O(g(x))$ means for sufficiently large $x$, $f(x)$ is bounded
by a constant times $g(x)$ in absolute value.




\section{The LLL Algorithm}

Consider lattices in the complex Euclidean space $\mathbb{C}^n$. A complex lattice is defined as
the set of points $L = \{\mathbf{Bx}|\mathbf{x}\in\mathcal{G}^n\}$, where $\mathbf{B}\in
\mathbb{C}^{n \times n}$ is referred to as the basis matrix, and
$\mathcal{G}=\mathbb{Z}+j\mathbb{Z}$, $j=\sqrt{-1}$ is the set of Gaussian integers. For
convenience, we only consider a square matrix $\mathbf{B}$ in this paper, while the extension to a
tall matrix is straightforward. Aside from the interests in digital communications
\cite{mow03:wcmc, ComplexLLL} and coding \cite{BK:Conway93}, complex lattices have found
applications in factoring polynomials over Gaussian integers \cite{Howgrave-GrahamThesis}.

A lattice $L$ can be generated by infinitely many bases, from which one would like to select one
that is in some sense nice or reduced. In many applications, it is advantageous to have the basis
vectors as short as possible. The LLL algorithm is a polynomial time algorithm that finds short
vectors within an exponential approximation factor \cite{LLL}. The complex LLL algorithm, which is
a modification of its real counterpart, has been described in \cite{Napias96,
Howgrave-GrahamThesis, mow03:wcmc, ComplexLLL}. The complex LLL algorithm works directly with
complex basis $\mathbf{B}$ rather than converting it into the real equivalent. Although the cost
for each complex arithmetic operation is higher than its real counterpart, the total number of
operations required for complex LLL reduction is approximately half of that for real LLL
\cite{ComplexLLL}.

\subsection{Gram-Schmidt (GS) orthogonalization (O), QR and Cholesky decompositions}

For a matrix $\mathbf{B} = [\mathbf{b}_1, ..., \mathbf{b}_n] \in \mathbb{C}^{n\times n}$, the
classic GSO is defined as follows \cite{golub}
\begin{equation}\label{GSO}
  \mathbf{\hat{b}}_i = \mathbf{b}_i -  \sum_{j=1}^{i-1}{\mu_{i,j}\mathbf{\hat{b}}_j},\quad \text{for } i = 1,...,n
\end{equation}
where $\mu_{i,j} = \langle \mathbf{b}_i, \mathbf{\hat{b}}_j\rangle/\|\mathbf{\hat{b}}_j\|^2$. In
matrix notation, it can be written as $\mathbf{B}=\mathbf{\hat{B}}\bm{\mu}^T$, where
$\mathbf{\hat{B}}=[\mathbf{\hat{b}}_1, ..., \mathbf{\hat{b}}_n]$, and $\bm{\mu}=[\mu_{i,j}]$ is a
lower-triangular matrix with unit diagonal elements.

GSO is closely related to QR decomposition $\mathbf{B}=\mathbf{QR}$, where $\mathbf{Q}$ is an
orthonormal matrix and $\mathbf{R}$ is an upper-triangular matrix with nonnegative diagonal
elements. More precisely, one has the relations $\mu_{j,i}=r_{i,j}/r_{i,i}$ and
$\mathbf{\hat{b}}_i=r_{i,i}\cdot \mathbf{q}_i$ where $\mathbf{q}_i$ is the $i$-th column of
$\mathbf{Q}$. QR decomposition can be implemented in various ways such as GSO, Householder and
Givens transformations \cite{golub}.

The Cholesky decomposition $\mathbf{A} = \mathbf{R}^H \mathbf{R}$ computes the R factor of the QR
decomposition from the Gram matrix $\mathbf{A}=\mathbf{B}^H \mathbf{B}$. Given the Gram matrix, the
computational complexity of Cholesky decomposition is approximately $n^3/3$, which is lower than
$2n^3$ of QR decomposition \cite{golub}.

\subsection{LLL Reduction}


\begin{deft}[Complex LLL]\label{DefinitionLLL}
Let $\mathbf{B}=\mathbf{\hat{B}}\bm{\mu}^T$ be the GSO of a complex-valued basis ${\bf B}$. ${\bf
B}$ is LLL-reduced if both of the following conditions are satisfied:
\begin{equation}
\label{lll-cond-3}
    |\Re(\mu_{i,j})| \leq 1/2 \text{ and } |\Im(\mu_{i,j})|\ \leq 1/2
\end{equation}
for $1 \leq j < i \leq n$, and
\begin{equation} \label{lll-cond-4}
    \|\mathbf{\hat{b}}_i\|^2 \geq (\delta - |\mu_{i,i-1}|^2) \|\mathbf{\hat{b}}_{i-1}\|^2
\end{equation}
for $1 < i \leq n$, where $1/2 < \delta \leq 1$ is a factor selected to achieve a good
quality-complexity tradeoff.
\end{deft}

The first condition is the size-reduced condition, while the second is known as the Lov\'asz
condition. It follows from the Lov\'asz condition (\ref{lll-cond-4}) that for an LLL-reduced basis
\begin{equation}\label{Lovasz1}
\|\mathbf{\hat{b}}_i\|^2 \geq (\delta - 1/2)\|\mathbf{\hat{b}}_{i-1}\|^2,
\end{equation}
i.e., the lengths of GS vectors do not drop too much.

Let $\alpha = 1/(\delta-1/2)$. A complex LLL-reduced basis satisfies \cite{Napias96,
Howgrave-GrahamThesis}:
\begin{equation}\label{LLL-bounds}
\begin{split}
    ||{\bf b}_1|| &\leq \alpha^{(n-1)/4}\text{det}^{1/n}L, \\
    ||{\bf b}_1|| &\leq \alpha^{(n-1)/2}\lambda_1,\\
    \prod_{i=1}^n||{\bf b}_i|| &\leq \alpha^{n(n-1)/4}\det L, \\
\end{split}
\end{equation}
where $\lambda_1$ is the length of the shortest vector in $L$, and $\det{L} \triangleq \det
\mathbf{B}$. These properties show in various senses that the vectors of a complex LLL-reduced
basis are not too long. Analogous properties hold for real LLL, with $\alpha$ replaced with $\beta
= 1/(\delta-1/4)$ \cite{LLL}. It is noteworthy that although the bounds (\ref{LLL-bounds}) for
complex LLL are in general weaker than the real-valued counterparts, the actual performances of
complex and real LLL algorithms are very close, especially when $n$ is not too large.

A size-reduced basis can be obtained by reducing each vector individually. The vector
$\mathbf{b}_k$ is size-reduced if $|\Re(\mu_{k,l})| \leq 1/2$ and $|\Im(\mu_{k,l})| \leq 1/2$ for
all $l < k$. Algorithm \ref{alg:sizereduction} shows how $\mathbf{b}_k$ is size-reduced against
$\mathbf{b}_{l}$ ($l<k$). To size-reduce $\mathbf{b}_k$, we call Algorithm \ref{alg:sizereduction}
for $l = k-1$ down to 1. Size-reducing $\mathbf{b}_k$ does not affect the size reduction of the
other vectors. Furthermore, it is not difficult to see that size reduction does not change the GS
vectors.

Algorithm \ref{alg:LLL} describes the LLL algorithm (see \cite{ComplexLLL} for the pseudo-code of
complex LLL). It computes a reduced basis by performing size reduction and swapping in an iterative
manner. If the Lov\'asz condition (\ref{lll-cond-4}) is violated, the basis vectors $\mathbf{b}_k$
and $\mathbf{b}_{k-1}$ are swapped; otherwise it carries out size reduction to satisfy
(\ref{lll-cond-3}). The algorithm is known to terminate in a finite number of iterations for any
given input basis $\mathbf{B}$ and for $\delta \leq 1$ \cite{LLL} (note that this is true even when
$\delta = 1$ \cite{Akhavi03, LenstraFlags}). By an iteration we mean the operations within the
``while" loop in Algorithm \ref{alg:LLL}, which correspond to an increment or decrement of the
variable $k$.


\begin{algorithm}[t]
\caption{\quad Pairwise Size Reduction} \label{alg:sizereduction}
{\bf{Input}}: Basis vectors $\mathbf{b}_k$ and $\mathbf{b}_l$ ($l<k$)\\
\text{\hspace{0.95cm}} GSO coefficient matrix $[\mu_{i,j}]$\\
{\bf{Output}}: Basis vector $\mathbf{b}_k$ size-reduced against $\mathbf{b}_l$ \\
\text{\hspace{1.2cm}} Updated GSO coefficient matrix $[\mu_{i,j}]$

\begin{algorithmic}[0]

\IF{$|\Re(\mu_{k,l})| \geq 1/2$ or $|\Im(\mu_{k,l})| \geq 1/2$}

\STATE ${\bf b}_k := {\bf b}_k - \lceil\mu_{k,l}\rfloor {\bf b}_{l}$


\FOR{$j=1,2,...,l$} \STATE $\mu_{k,j} := \mu_{k,j} - \lceil\mu_{k,l}\rfloor\mu_{l,j}$ \ENDFOR


\ENDIF

\end{algorithmic}
\end{algorithm}

\begin{algorithm}[t]
\caption{\quad LLL Algorithm} \label{alg:LLL}
{\bf{Input}:} A basis $\mathbf{B} = [\mathbf{b}_1, ... \mathbf{b}_n]$\\
{\bf{Output}:} The LLL-reduced basis
\begin{algorithmic}[1]

\STATE compute GSO $\mathbf{B}=\mathbf{\hat{B}} [\mu_{i,j}]^T$

\STATE $k := 2$

\WHILE{$k \leq n$}

\STATE size-reduce $\mathbf{b}_k$ against $\mathbf{b}_{k-1}$

\IF{$\|\mathbf{\hat{b}}_k + \mu_{k,k-1}\mathbf{\hat{b}}_{k-1}\|^2 < \delta
  \|\mathbf{\hat{b}}_{k-1}\|^2$}

\STATE swap $\mathbf{b}_k$ and $\mathbf{b}_{k-1}$ and update GSO

\STATE $k := \max(k-1,2)$

\ELSE

\FOR{$l = k-2, k-3, ..., 1$} \STATE size-reduce $\mathbf{b}_k$ against $\mathbf{b}_l$ \ENDFOR

\STATE $k := k + 1$

\ENDIF

\ENDWHILE
\end{algorithmic}
\end{algorithm}

Obviously, for real-valued basis matrix $\mathbf{B}$, Definition \ref{DefinitionLLL} and Algorithm
\ref{alg:LLL} coincide with the standard real LLL algorithm. Some further relations between real
and complex LLL are discussed in Appendix \ref{AppendixI}.

\begin{rem}
The LLL algorithm can also operate on the Gram matrix \cite{Nguyen2}. To do this, one applies the
Cholesky decomposition and updates the Gram matrix. Everything else remains pretty much the same.
\end{rem}

\subsection{LLL-Deep}

LLL-deep extends the swapping step to all vectors before $\mathbf{b}_k$, as shown in Algorithm
\ref{alg:LLLdeep}. The standard LLL algorithm is restricted to $i=k-1$ in Line 5. LLL-deep can find
shorter vectors. However, there are no proven bounds for LLL-deep other than those for standard
LLL. The experimental complexity of LLL-deep is a few times as much as that of LLL, although the
worst-case complexity is exponential. To limit the complexity, it is common to restrict the
insertion within a window \cite{SchnorrEuchner}. However, we will not consider this window in this
paper.




%
%
%
%
%
%
%

\begin{algorithm}[t]
\caption{\quad LLL Algorithm with Deep Insertion} \label{alg:LLLdeep}
{\bf{Input}:} A basis $\mathbf{B} = [\mathbf{b}_1, ... \mathbf{b}_n]$\\
{\bf{Output}:} The LLL-deep-reduced basis
\begin{algorithmic}[1]

\STATE compute GSO $\mathbf{B}=\mathbf{\hat{B}} [\mu_{i,j}]^T$

\STATE $k := 2$

\WHILE{$k \leq n$}

\STATE size-reduce $\mathbf{b}_k$ against $\mathbf{b}_{k-1}, ..., \mathbf{b}_{2}, \mathbf{b}_{1}$

\IF{$\exists i, 1\leq i <k$ such that
$\sum_{j=i}^k\mu_{k,j}^2\|\mathbf{\hat{b}}_j\|^2<\delta\|\mathbf{\hat{b}}_i\|^2$}

\STATE for the smallest such $i$, insert $\mathbf{b}_k$ before $\mathbf{b}_{i}$ and update GSO

\STATE $k := \max(i,2)$

\ELSE

\STATE $k := k + 1$

\ENDIF

\ENDWHILE
\end{algorithmic}
\end{algorithm}

\section{Complexity Analysis of Complex LLL}

In the previous work \cite{ComplexLLL}, it was only qualitatively argued that complex LLL
approximately reduces the complexity by half, while a rigorous analysis was lacking. In this
section, we complement the work in \cite{ComplexLLL} by evaluating the computational complexity in
terms of (complex-valued) floating-point operations (flops). The other operations such as looping
and swapping are ignored. The complexity analysis consists of two steps. Firstly, we bound the
average number of iterations. Secondly, we bound the number of flops of a single iteration.

\subsection{Average Number of Iterations}
To analyze the number of iterations, we use a standard argument, where we consider the LLL
potential \cite{LLL}
\begin{equation}\label{potential}
  \mathcal{D} = \prod_{i=1}^{n-1}{\|\mathbf{\hat{b}}_i\|^{2(n-i)}}.
\end{equation}
Obviously, $\mathcal{D}$ only changes during the swapping step. This happens when
\begin{equation} \label{swap}
    \|\mathbf{\hat{b}}_k\|^2 < (\delta - |\mu_{k,k-1}|^2)\|\mathbf{\hat{b}}_{k-1}\|^2
\end{equation}
for some $k$. After swapping, $\mathbf{\hat{b}}_{k-1}$ is replaced by
$\mathbf{\hat{b}}_{k}+\mu_{k,k-1}\mathbf{\hat{b}}_{k-1}$. Thus $\|\mathbf{\hat{b}}_{k-1}\|^2$ as
well as $\mathcal{D}$ shrinks by a factor less than $\delta$.

The number $K$ of iterations is exactly the number of Lov\'asz tests. Let $K^+$ and $K^-$ be the
numbers of positive and negative tests, respectively. Obviously, $K=K^++K^-$. Since $k$ is
incremented in a positive test and decremented in a negative test, and since $k$ starts at 2 and
ends at $n$, we must have $K^+ \leq K^-+(n-1)$ (see also \cite{Daude}). Thus it is sufficient to
bound $K^-$.

Let $A = \max_i{\|\mathbf{\hat{b}}_i\|^2}$ and $a = \min_i{\|\mathbf{\hat{b}}_i\|^2}$. The initial
value of $\mathcal{D}$ can be bounded from above by $A^{n(n-1)/2}$. To bound the number of
iterations for a complex-valued basis, we invoke the following lemma \cite{LLL,Daude}, which holds
for complex LLL as well.

\begin{lem}\label{Lemma1}
During the execution of the LLL algorithm, the maximum $A$ is non-increasing while the minimum $a$
is non-decreasing.
\end{lem}

In other words, the LLL algorithm tends to reduce the interval $[a,A]$ where the squared lengths of
GS vectors reside. From Lemma \ref{Lemma1}, we obtain
\begin{equation}\label{num_iter}
  K^- \leq \frac{n(n-1)}{2}\log\frac{A}{a}
\end{equation}
where the logarithm is taken to the base $1/\delta$ (this will be the case throughout the paper).

Assuming that the basis vectors are i.i.d. in the unit ball of $\mathbb{R}^n$, Daude and Vallee
showed that the mean of $K^-$ is upper-bounded by $O(n^2\log n)$ \cite{Daude}. The analysis for the
i.i.d. Gaussian model is similar. Yet, here we use the exact value of $\mathcal{D}_0$, which is the
initial value of $\mathcal{D}$, to bound the average number of iterations. It leads to a better
bound than using the maximum $A$ in (\ref{num_iter}). The lower bound on $\mathcal{D}$ is followed:
\[\mathcal{D}_\text{lower} = \frac{n(n-1)}{2}\log a \leq \mathcal{D}.\]
Accordingly, the mean of $K^-$ is bounded by
\begin{equation}\label{Kminus}
\begin{split}
  E[K^-] &\leq E \left[ \log \frac{\mathcal{D}_0}{\mathcal{D}_\text{lower}}\right]\\
  &= E\left[\log \mathcal{D}_0\right] - E\left[\log
  \mathcal{D}_\text{lower}\right]\\
  &= E\left[\log \mathcal{D}_0\right] - \frac{n(n-1)}{2}\log E\left[\log
  a\right].
\end{split}
\end{equation}
We shall bound the two terms separately.

The QR decomposition of an i.i.d. complex normal random matrix has the following property: the
squares of the diagonal elements $r_{i,i}$ of the matrix $\mathbf{R}$ are statistically independent
$\chi^2$ random variables with $2(n-i+1)$ degrees of freedom \cite{Goodman}. Since
$r^2_{i,i}=\|\mathbf{\hat{b}}_i\|^2$, we have
\begin{equation}\label{D0}
\begin{split}
  E\left[\log \mathcal{D}_0\right] &= \sum_{i=1}^{n-1}{(n-i)E\left[\log
  \|\mathbf{\hat{b}}_i\|^2\right]}\\
  & \leq \sum_{i=1}^{n-1}{(n-i)\log
  E\left[\|\mathbf{\hat{b}}_i\|^2\right]}\\
  & = \sum_{i=1}^{n-1}{(n-i)\log2(n-i+1)}\\
  & \leq \frac{n(n-1)}{2}\log 2n\\
\end{split}
\end{equation}
where the first inequality follows from Jensen's inequality $E[\log X]\leq \log E[X]$.

It remains to determine $E\left[\log a\right]$. The cumulative distribution function (cdf) $F_a(x)$
of the minimum
$a$ can be written as 
\[F_a(x)=1-\prod_{i=1}^{n} \left[1-F_i(x)\right]\]
where $F_i(x)$ denotes the cdf of a $\chi^2$ random variables with $2i$ degrees of freedom:
\begin{equation}\label{chi2cdf}
\begin{split}
  F_i(x) &= \int_0^x{\frac{1}{2^{i}\Gamma(i)}y^{i-1}e^{-y/2}dy}, \quad x \geq 0\\
  & = 1 - e^{-x} \sum_{m=0}^{i-1} {\frac{x^m}{m!}}.
\end{split}
\end{equation}
As $n$ tends to infinity, $F_a(x)$ approaches a limit cdf $\bar{F}_a(x)$ that is not a function of
$n$. Since $a$ deceases with $n$, $E\left[\log a\right]$ is necessarily bounded from below by its
limit:
\begin{equation}\label{integratmin}
  E\left[\log a\right] \geq \int_0^\infty \log x d\bar{F}_a(x).
\end{equation}
The convergence of $F_a(x)$ is demonstrated in Fig. \ref{fig:chi2min}.
\begin{figure}[t]

\centering\centerline{\epsfig{figure=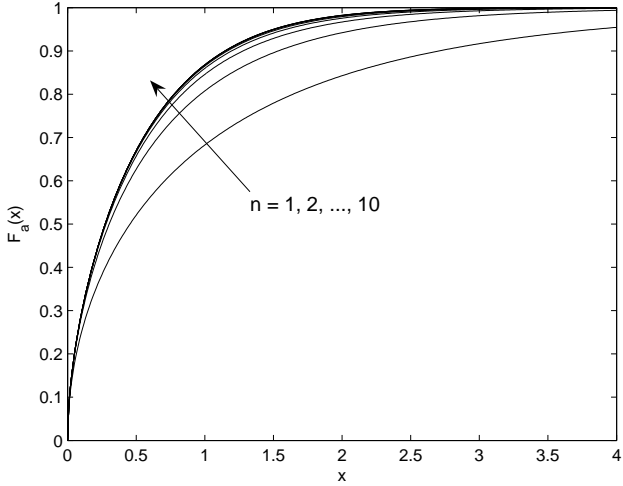,width=9.4cm}}

\caption{Convergence of $F_a(x)$ as $n$ increases.}

\label{fig:chi2min}
\end{figure}

Although it is difficult to evaluate the above integral exactly, we can derive a lower bound. To do
this, we examine the behavior of the limit probability density function (pdf) $\bar{f}_a(x)$.
$\bar{f}_a(x)$ is a decreasing function. Moreover, $\bar{f}_a(0)=1$. To show this, note that as
$x\rightarrow 0^+$ we have the approximation $e^{-x} \approx 1-x$ and $\sum_{m=0}^{i-1} x^m/m!
\approx 1+x $. Therefore, as $x\rightarrow 0^+$
\begin{equation}\label{approx}
\begin{split}
  \bar{F}_a(x) &\approx 1 - \lim_{n\rightarrow\infty}(1-x)\prod_{i=2}^n{(1-x)(1+x)}\\
  & = 1 - \lim_{n\rightarrow\infty}(1-x)(1-x^2)^{n-1}\\
\end{split}
\end{equation}
and accordingly $\bar{F}_a(x) \approx x$ as $x \rightarrow 0^+$. Then
$\bar{f}_a(0)=d\bar{F}_a(x)/dx|_{x=0^+} = 1$.

Then we have
\begin{equation}\label{Eloga}
\begin{split}
  E\left[\log a\right] &\geq \int_0^\infty \log x \bar{f}_a(x)dx\\
  & \geq \int_0^1 \log x \bar{f}_a(x)dx\\
  & \geq \int_0^1 \log x dx = -1
\end{split}
\end{equation}
where the last inequality follows from $\bar{f}(x) < 1$ for $x>0$.

Substituting (\ref{D0}),(\ref{Eloga}) into (\ref{Kminus}), we obtain
\begin{equation}\label{KminusBound}
    E[K^-] \leq \frac{n(n-1)}{2}(\log 2n + 1).
\end{equation}
Therefore, we arrive at the following result:

\begin{prop}
Under the i.i.d. complex normal model of the basis matrix $\mathbf{B}$, the total number of
iterations of the LLL algorithm can be bounded as
\begin{equation}\label{KBound}
    E[K] \leq {n(n-1)}(\log 2n + 1)+n \approx n^2\log n.
\end{equation}
\end{prop}

\text{ }

\begin{rem}
A similar analysis in \cite{Jalden08} applied the bounds $\sigma_1 \geq \sqrt{A}$ and $\sqrt{a} \geq
\sigma_n$, where $\sigma_1$ and $\sigma_n$ are the maximum and minimum singular value of
$\mathbf{B}$, respectively. Accordingly, the resultant bound $n^2\log(\sigma_1/\sigma_n)$ is less
tight. In fact, \cite{Jalden08} showed the bound $E[K] \lesssim 4n^2\log n$, which is larger by a
factor of 4.
\end{rem}

\subsection{Number of Flops for Each Iteration}


The second step proceeds as follows. Updating the GSO coefficients \cite{LLL} during the swapping
step costs $6(n-k)+7 \leq 6n-5$ flops ($k\geq 2$), whereas pairwise size reduction for $(k,k-1)$
costs $2n+2(k-1) \leq 4n-2$ flops ($k\geq 2$). Testing the Lov\'asz condition as (\ref{Lovasz1})
costs 3 flops each time. Besides, the initial GSO costs $2n^3$ flops. Therefore, excluding full
size reduction, the cost is bounded by\footnote{The reason why we separately counts $\mathcal{C}_1$
and $\mathcal{C}_2$ will become clear in the next Section.}
\begin{equation}\label{totalcost1}
\begin{split}
  \mathcal{C}_1 &\leq (6n-5)K^- + (4n-2+3)(K^-+K^+) + 2n^3\\
              &\leq (6n-5)\frac{n(n-1)}{2}\log 2n\\
              &\quad + (4n+1)[n(n-1)\log 2n +(n-1)]+2n^3\\
              &= 7n^2(n-1)\log 2n -\frac{3}{2}n(n-1)\log 2n \\
              &\quad+(4n+1)(n-1)+2n^3\\
              &\leq 7n^3\log 2n+2n^3.
\end{split}
\end{equation}

During each step, the number of flops due to full size reduction is no more than
\begin{equation}
{\sum_{l=1}^{k-2}{(2n+2l)}} \leq 3n^2.
\end{equation}
Therefore, the subtotal amount of flops due to full size reduction are
\begin{equation}\label{extracost}
  \mathcal{C}_2 = 3n^2K^+ \leq 3n^2 \left[\frac{n(n-1)}{2}\log 2n + (n-1)\right]
\end{equation}
which is $O(n^4 \log n)$ and thus is dominant. This results in $O(n^4 \log n)$ complexity bound on
the overall complexity $\mathcal{C}=\mathcal{C}_1+\mathcal{C}_1$ of the complex LLL algorithm.

%
%
%
%

\subsection{Comparison with Real-valued LLL}

In the conventional approach, the complex-valued matrix $\mathbf{B}$ is converted into real-valued
$\mathbf{B}_{\text{R}}$ (see Appendix \ref{AppendixI}). Obviously, $\mathbf{B}$ and
$\mathbf{B}_{\text{R}}$ have the same values of $A$ and $a$, and same convergence speed determined
by $\delta$. Since the size of $\mathbf{B}_{\text{R}}$ is twice that of $\mathbf{B}$, real LLL
needs about four times as many iterations as complex LLL.

For real-valued LLL, the expressions (\ref{totalcost1}) and (\ref{extracost}) are almost the same
\cite{LingISIT07}, with a doubled size $n$. Under the assumption that on average a complex
arithmetic operation costs four times as much as a real operation, $\mathcal{C}_1$ of complex LLL
is half of that of real LLL. Meanwhile, when comparing $\mathcal{C}_2$, there is a subtle
difference, i.e., the chance of full size reduction (Lines 9 and 10 in Algorithm \ref{alg:LLL}) is
doubled for complex LLL \cite{ComplexLLL}. Therefore, $\mathcal{C}_2$ of complex LLL is also half.
Then the total cost of complex LLL is approximately half.

\subsection{Reduction of the Dual Basis}


Sometimes it might be more preferable to reduce the dual basis $\mathbf{B}^* \triangleq
(\mathbf{B}^{-1})^H\mathbf{J}$, where $\mathbf{J}$ is the column-reversing matrix \cite{LingIT07}.
In the following, we show that the $O(n^2\log n)$ average number of iterations still holds.

Let $\mathbf{\hat{B}}^*, A^*, a^*$ be the corresponding notations for the dual basis. Due to the
relation $\|\hat{\mathbf{b}}_{i}\|=\|\hat{\mathbf{b}}_{n-i+1}^*\|^{-1}$ \cite{lagarias}, we have
\[\frac{A}{a}=\frac{1/a^*}{1/A^*}=\frac{A^*}{a^*}.\] Thus, the bound on the number of iterations is
\emph{exactly the same}. In particular, the average number of iterations is the same.

In MIMO broadcast, it is $\mathbf{B}^{-1}$ that needs to be reduced \cite{Windpassinger04}. Noting
that $\mathbf{B}^{-1}$ has the same statistics as $(\mathbf{B}^{-1})^H\mathbf{J}$ because
$\mathbf{B}$ is i.i.d. normal, it is easy to see that the bound on $K^-$ is again the same.

\begin{rem}
The same conclusion was drawn in \cite{Jalden08} by examining the singular values.
\end{rem}

%
%
%
%
%
%
%
%
%


\section{Effective LLL Reduction}

\subsection{Effective LLL}

Since some applications such as sphere decoding and SIC only require the GS vectors
$\mathbf{\hat{b}}_{i}$ rather than the basis vectors themselves, and since size reduction does not
change them, a weaker version of the LLL algorithm is sufficient for such applications. This makes
it possible to devise a variant of the LLL algorithm that has lower theoretic complexity than the
standard one.

From the above argument it seems that we would be able to remove the size reduction operations at
all. However, this is not the case. An inspection shows that the size-reduced condition for two
consecutive basis vectors
\begin{equation}\label{EffectiveCondition}
  |\Re(\mu_{i,i-1})| \leq 1/2 \text{ and } |\Im(\mu_{i,i-1})|\ \leq 1/2, \quad 1 < i \leq n
\end{equation}
is essential in maintaining the lengths of the GS vectors. In other words,
(\ref{EffectiveCondition}) must be kept along with the Lov\'asz condition so that the lengths of GS
vectors will not be too short. Note that their lengths are related to the performance of SIC and
the complexity of sphere decoding. We want the lengths to be as even as possible so as to improve
the SIC performance and reduce the complexity of sphere decoding.

A basis satisfies condition (\ref{EffectiveCondition}) and the Lov\'asz condition
(\ref{lll-cond-4}) is called an {\it effectively} LLL-reduced basis in \cite{Howgrave-Graham97}.
Effective LLL reduction terminates in exactly the same number of iterations, because size-reducing
against other vectors has no impact on the Lov\'asz test. In addition to (\ref{Lovasz1}), an
effectively LLL-reduced basis has other nice properties. For example, if a basis is effectively
LLL-reduced, so is its dual basis \cite{Howgrave-Graham97,Howgrave-GrahamThesis}.

Effective LLL reduction permits us to remove from Algorithm \ref{alg:LLL} the most expensive part,
i.e., size-reducing $\mathbf{b}_k$ against $\mathbf{b}_{k-2}, \mathbf{b}_{k-3}, ..., \mathbf{b}_1$
(Lines 9-10). For integral bases, doing this may cause excessive growth of the (rational) GSO
coefficients $\mu_{i,j}$, $j<i-1$, and the increase of bit lengths will likely offset the
computational saving. This is nonetheless not a problem in MIMO decoding, since the basis vectors
and GSO coefficients can be represented by floating-point numbers after all. We use a model where
floating-point operations take constant time, and accuracy is assumed not to perish. There is
strong evidence that this model is practical, because the correctness of floating-point LLL for
\emph{integer} lattices has been proven \cite{Nguyen2}. Although the extension of the proof to the
case of continuous bases seems very difficult, in practice this model is valid as long as the
arithmetic precision is sufficient for the lattice dimensions under consideration.


We emphasize that under this condition the effective and standard LLL algorithms have the same
error performance in the application to SIC and sphere decoding, as asserted by Proposition
\ref{prop1}.

\begin{prop}\label{prop1}
The SIC and sphere decoder with effective LLL reduction finds {\it exactly the same} lattice point
as that with standard LLL reduction.
\end{prop}

\begin{proof}
This is obvious since SIC and sphere decoding only need the GS vectors and since standard and
effective LLL give exactly the same GS vectors.
\end{proof}

\subsection{Transformation to Standard LLL-Reduced Basis}


On the other hand, ZF does require the condition of full size reduction. One can easily transform
an effective LLL-reduced basis into a fully reduced one. To do so, we simply perform size
reductions at the end to make the other coefficients $|\Re(\mu_{i,j})| \leq 1/2$ and
$|\Im(\mu_{i,j})| \leq 1/2$, for $1\leq j<i-1$, $2<i\leq n$
\cite{Storjohann,Howgrave-GrahamThesis}. This is because, once again, such operations have no
impact on the Lov\'asz condition. Full size reduction costs $O(n^3)$ arithmetic operations. The
analysis in the following subsection will show the complexity of this version of LLL reduction is
on the same order of that of effective LLL reduction. In other words, it has lower theoretic
complexity than the standard LLL algorithm.

There are likely multiple bases of the lattice $L$ that are LLL-reduced. For example, a basis
reduced in the sense of Korkin-Zolotarev (KZ) is also LLL-reduced. Proposition \ref{prop2} shows
that this version results in the same reduced basis as the LLL algorithm.

\begin{prop}\label{prop2}
Fully size-reducing an effectively LLL-reduced basis gives {\it exactly the same} basis as the
standard LLL algorithm.
\end{prop}

\begin{proof}
It is sufficient to prove the GSO coefficient matrix $[\mu_{i,j}]$ is the same, since
$\mathbf{B}=\mathbf{\hat{B}}[\mu_{i,j}]^T$ and since $\mathbf{\hat{B}}$ is not changed by size
reduction. We prove it by induction. Suppose the new version has the same coefficients $\mu_{i,j}$,
$j<i$ when $i=2, ..., k-1$. Note that this is obviously true when $i=2$.

When $i=k$ and $l=k-2$, the new version makes $|\Re(\mu_{k,k-2})|<1/2$ and $|\Im(\mu_{k,k-2})|<1/2$
at the end by subtracting its integral part so that
\[|\Re(\mu_{k,k-2}-\lceil\mu_{k,k-2}\rfloor)|<1/2, \text{ } |\Im(\mu_{k,k-2}-\lceil\mu_{k,k-2}\rfloor)|<1/2.\] The standard LLL algorithm achieves this in a number of
iterations. Yet the sum of integers subtracted must be equal. The other coefficients will be
updated as \[\mu_{k,j} := \mu_{k,j} - \lceil\mu_{k,k-2}\rfloor\mu_{k-2,j}, \quad j=1,...,k-3,\]
which will remain the same since coefficients $\mu_{k-2,j}$ are assumed to be the same.

Clearly, the argument can be extended to the case $i=k$ and $l=k-3, \cdots, 1$. That is, the new
version also has the same coefficients $\mu_{k,j}$, $j<k$. This completes the proof.
\end{proof}

\subsection{$O(n^3 \log n)$ Complexity}








\begin{prop}
Under the i.i.d. complex normal model of the basis $\mathbf{B}$, the average complexity of
effective LLL is bounded by $\mathcal{C}_1$ in (\ref{totalcost1}) which is $O(n^3\log n)$.
\end{prop}

\begin{proof}
Since the number of iterations is the same, and since each iteration of effective LLL costs $O(n)$
arithmetic operations, the total computation cost is $O(n^3\log n)$. More precisely, the effective
LLL consists of the following computations: initial GSO, updating the GSO coefficients during the
swapping step, pairwise size reduction, and testing the Lov\'asz condition. Therefore, the total
cost is exactly bounded by $\mathcal{C}_1$ in (\ref{totalcost1}).
\end{proof}

To obtain a fully reduced basis, we further run pairwise size reduction for $l = k-2$ down to 1 for
each $k = 3,\cdots,n$. The additional number of flops required is bounded by
\begin{equation}\label{SRcost}
\begin{split}
  \sum_{k=3}^n{\sum_{l=1}^{k-2}{(2n+2l)}} &= \sum_{k=3}^n{{[2n(k-2)+(k-1)(k-2)]}}\\
  &=\frac{4}{3}n(n-1)(n-2) \leq \frac{4}{3}n^3.
\end{split}
\end{equation}
Obviously, the average complexity is still $O(n^3\log n)$.

Again, since each complex arithmetic operation on average requires four real arithmetic operations,
the net saving in complexity due to complex effective LLL is about 50\%.




In Fig. \ref{fig:ELLL-time}, we show the theoretic upper bounds and experimental results for
effective and standard LLL algorithms. Clearly there is room to improve the analysis. This is
because our work is in fact a blend of worst- and average-case analysis, and the resultant
theoretic bound is unlikely to be sharp. But nonetheless, the experimental data exhibit cubic
growth with $n$, thereby supporting the $O(n^3 \log n)$ bound. On the other hand, surprisingly, the
experimental complexity of standard LLL reduction is not much higher than that of effective LLL. We
observed that this is because of the small probability to execute size reduction with respect to
nonconsecutive vectors (Lines 9-10 of the standard LLL algorithm), which were thought to dominate
the complexity.



\begin{figure}[t]
\centering\centerline{\epsfig{figure=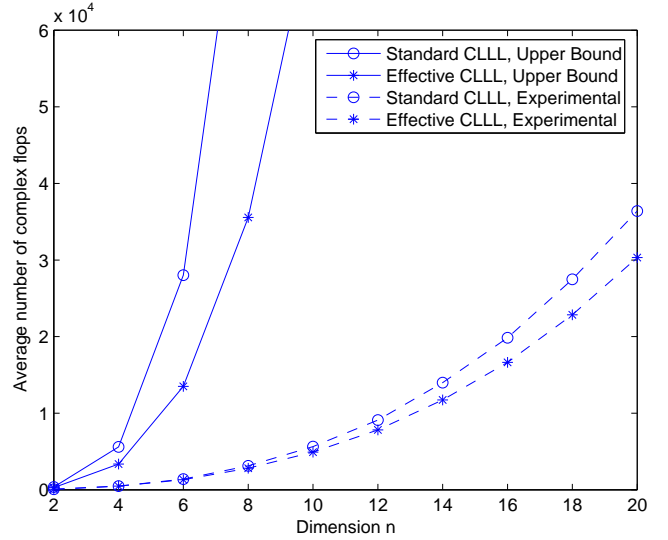,width=10cm}}

\caption{Average number of complex flops for effective CLLL reduction with $\delta = 3/4$ for the
i.i.d. normal basis model.}


\label{fig:ELLL-time}
\end{figure}

%
%

\section{Fixed-Complexity Implementation}

Although the average complexity of the LLL algorithm is polynomial, it is important to recognize
that its complexity is in fact variable due to its sequential nature. The worst-case complexity
could be very large \cite{Jalden08}, which may severely limit the throughput of the decoder. In
this Section, we propose two fixed-complexity structures that are suitable to approximately
implement the LLL algorithm in hardware. One is fixed-complexity effective LLL, while the other is
fixed-complexity LLL-deep. The structures are based on parallel versions of the LLL and LLL-deep
algorithms, respectively. The parallel versions exhibit fast convergence, which allow to fix the
number of iterations without incurring much quality degradation.

\subsection{Fixed-Complexity Effective LLL}

Recently, a fixed-complexity LLL algorithm was proposed in \cite{VetterLLL}. It resembles the
parallel `even-odd' LLL algorithm earlier proposed in \cite{Villard} (see also \cite{Wang09-EOLLL}
for a systolic-array implementation) and a similar algorithm in \cite{Mow:Thesis}. It is well known
that the LLL reduction can be achieved by performing size reduction and swapping in any order.
Therefore, the idea in \cite{VetterLLL} was to run a super-iteration where the index $k$ is
monotonically incremented from 2 to $n$, and repeat this super-iteration until the basis is
reduced. This is slightly different from the `even-odd' LLL \cite{Villard}, where the
super-iteration is performed for even and odd indexes $k$ separately.

Here, we extend this fixed-complexity structure to effective LLL, which is a truncated version of
the parallel effective LLL described in Algorithm \ref{alg:FixedELLL} (one can easily imagine an
`even-odd' version of this algorithm). The index $k$ is never reduced in a super-iteration. Of
course, one can further run full size reduction to make the basis reduced in the sense of LLL. It
is easy to see that this algorithm converges by examining the LLL potential function. To cater for
fixed-complexity implementation, we run a sufficiently large but fixed number of super-iterations.

How many super-iterations should we run? A crude estimate is $O(n^2\log n)$, i.e., the same as that
for standard LLL. Since there is at least one swap within a super-iteration (otherwise the
algorithm terminates), the number of super-iterations is bounded by $O(n^2\log n)$ (this is the
approach used in \cite{VetterLLL}). However, this approach might be pessimistic, as up to $n-1$
swaps may occur in one super-iteration. Next, we shall argue that on average it is sufficient to
run $O(n\log n)$ super-iterations in order to obtain a good basis. Accordingly, since each
super-iteration costs $O(n^2)$, the overall complexity is $O(n^3 \log n)$.

\begin{algorithm}[t]
\caption{\quad Parallel Effective LLL Algorithm} \label{alg:FixedELLL}
{\bf{Input}:} A basis $\mathbf{B} = [\mathbf{b}_1, ... \mathbf{b}_n]$\\
{\bf{Output}:} Effectively LLL-reduced basis
\begin{algorithmic}[1]

\STATE compute GSO $\mathbf{B}=\mathbf{\hat{B}} [\mu_{i,j}]^T$

\WHILE{any swap is possible}

\FOR{$k = 2, 3, ...,n$}

\STATE size-reduce $\mathbf{b}_k$ against $\mathbf{b}_{k-1}$

\IF{$\|\mathbf{\hat{b}}_k + \mu_{k,k-1}\mathbf{\hat{b}}_{k-1}\|^2 < \delta
  \|\mathbf{\hat{b}}_{k-1}\|^2$}

\STATE swap $\mathbf{b}_k$ and $\mathbf{b}_{k-1}$ and update GSO

\ENDIF

\ENDFOR

\ENDWHILE
\end{algorithmic}
\end{algorithm}

\begin{prop}
Let $c=\frac{1}{\delta^2(\delta-\frac{1}{2})}$ for complex LLL and
$c=\frac{1}{\delta^2(\delta-\frac{1}{4})}$ for real LLL. On average the fixed-complexity effective
LLL finds a short vector with length
\begin{equation}\label{PELLL-b1}
    \|\mathbf{b}_1\| \leq \frac{c^{(n-1)/4}}{\sqrt{\delta}}\text{det}^{1/n}L
\end{equation}
after $O(n\log n)$ super-iterations.
\end{prop}


\begin{proof}
The proof is an extension of the analysis of `even-odd' LLL in \cite{Heckler}, with two
modifications.

Following \cite{Heckler}, define the ratios
\begin{equation}\label{ratiov}
    v(i) = \frac{\prod_{j=1}^{i}{\|\hat{\mathbf{b}}_j\|^2}}{c^{\frac{i(n-i)}{2}}(\det
    L)^{\frac{2i}{n}}}.
\end{equation}
Let $v_{\text{max}}=\max{\{v(i), 1\leq i \leq n\}}$. Following \cite{Heckler} one can prove that if
$v_{\max}>\frac{1}{\delta}$ and $v(i)>\delta v_{\max}$, then the swapping condition is satisfied,
i.e., $\|\hat{\mathbf{b}}_{i+1}\|^2 \leq (\delta-|\mu_{i+1,i}|^2) \|\hat{\mathbf{b}}_i\|^2$. Thus,
after swapping, any such $v(i)$ will be decreased by a factor less than $\delta$; all other ratios
do not increase. Hence $v_{\max}$ will be decreased by a factor less than $\delta$.

The first modification is bounding $v_{\max}$ in the beginning. We can see that in the beginning
$v_{\max}\leq A^n/\text{det}^{2}L \leq (A/a)^n$ (recall $A=\max_{i=1}^{n} \|\hat{\mathbf{b}}_i\|^2$
and $a=\min_{i=1}^{n} \|\hat{\mathbf{b}}_i\|^2$)\footnote{The bound $v_{\max}\leq A^n$ used in
\cite{Heckler} does not necessarily hold here since $\det L$ may be less than 1 for real (complex)
valued bases.}.

Secondly, we need to check whether the swapping condition $\|\hat{\mathbf{b}}_{i+1}\|^2 \leq
(\delta-|\mu_{i+1,i}|^2) \|\hat{\mathbf{b}}_i\|^2$ remains satisfied after $k$ goes from 2 to $i$.
It turns out to be true. This is because $\|\hat{\mathbf{b}}_i\|^2$ will not decrease; in fact, the
new $\|\hat{\mathbf{b}}_i\|^2$ increases by a factor larger than $1/\delta$ if ${\mathbf{b}}_{i-1}$
and ${\mathbf{b}}_i$ are swapped. Thus, ${\mathbf{b}}_{i}$ and ${\mathbf{b}}_{i+1}$ will always be
swapped regardless of the swap between ${\mathbf{b}}_{i-1}$ and ${\mathbf{b}}_i$, and accordingly
$v(i)$ will be decreased.

Thus, one has $v_{\max}\leq \frac{1}{\delta}$ after $O\left(n\log A - 2\log(\det L)\right) \leq
O(n\log A/a)$ iterations. On average this is $O(n\log n)$.

In particular, $v(1)\leq 1/ \delta$ implies that $\mathbf{b}_1$ is short, with length bounded in
(\ref{PELLL-b1}).
\end{proof}

\begin{rem} This analysis also applies to the fixed-complexity LLL proposed in \cite{VetterLLL} (this is obvious since
size reduction does not change GSO).
\end{rem}

\begin{rem} Compared to (\ref{LLL-bounds}) for standard sequential LLL, the approximation factor in
(\ref{PELLL-b1}) is larger by a coefficient $(\frac{1}{\delta})^{n/2}$. Yet we can make this factor
very small by choosing $\delta \lessapprox 1$. For example, if $\delta = 0.99$, then
$(\frac{1}{\delta})^{n/2} < 1.5$ for $n$ up to 80.
\end{rem}

\begin{rem} In practice, one can obtain a good basis after $n$ super-iterations. This will be
confirmed by simulation later in this section.
\end{rem}

\begin{rem}
Although we have only bounded the length of $\mathbf{b}_1$, this suffices for some applications in
lattice decoding. For example, the embedding technique (a.k.a. augmented lattice reduction
\cite{Luzzi10}) only requires a bound for the shortest vector.
\end{rem}

%
%
%
%

\subsubsection{Relation to PMLLL} Another variant of the LLL algorithm, PMLLL, was proposed in \cite{Mow:Thesis}, which repeats
two steps: one is a series of swapping to satisfy the Lov\'asz condition (even forgetting about
$|\mu_{k,k-1}|\leq 1/2$), the other is size reduction to make $|\mu_{k,k-1}|\leq 1/2$ for
$k=2,...,n$. This variant is similar to parallel effective LLL. However, $k$ does not necessarily
scan from $2$ to $n$ monotonically in PMLLL.

\subsection{Fixed-Complexity LLL-Deep}

Here, we propose another fixed-complexity structure to approximately implement LLL-deep (and,
accordingly, LLL). More precisely, we apply sorted GSO and size reduction alternatively. This
structure is closely related to V-BLAST.

The sorted GSO relies on the modified GSO \cite{golub}. At each step, the remaining columns of
$\mathbf{B}$ are projected onto the orthogonal complement of the linear space spanned by the GS
vectors already obtained. In sorted GSO, one picks the shortest GS vector at each step, which
corresponds to the sorted QR decomposition proposed by Wubben et al \cite{wubben}\footnote{Note
that this is contrary to the well known pivoting strategy where the longest Gram-Schmidt vector is
picked at each step so that $\|\mathbf{\hat{b}}_1\| \geq \|\mathbf{\hat{b}}_2\| \geq \cdots \geq
\|\mathbf{\hat{b}}_n\|$ \cite{golub}.} (we will use the terms sorted GSO and sorted QR
decomposition interchangeably). Algorithm {\ref{alg:SQR}} describes the process of sorted GSO.

For $i=1,...,n$, let $\pi_i$ denote the projection onto the orthogonal complement of the subspace
spanned by vectors $\mathbf{b}_1,...,\mathbf{b}_{i-1}$. Then the sorted GSO has the following
property:
$\mathbf{b}_1$ is the shortest vector among $\mathbf{b}_1$, $\mathbf{b}_2$, ..., $\mathbf{b}_n$;
$\pi_2(\mathbf{b}_2)$ is the shortest among $\pi_2(\mathbf{b}_2)$, ..., $\pi_2(\mathbf{b}_n)$; and
so on~\footnote{However, it is worth pointing out that,  sorted Gram-Schmidt orthogonalization does
not guarantee $\|\mathbf{\hat{b}}_1\| \leq \|\mathbf{\hat{b}}_2\| \leq \cdots \leq
\|\mathbf{\hat{b}}_n\|$. It is only a greedy algorithm that hopefully makes the first few
Gram-Schmidt vectors not too long, and accordingly, the last few not too short. The term ``sorted"
is probably imprecise, because the Gram-Schmidt vectors are not sorted in length at all.}.

Sorted GSO tends to reduce $\max\{\|\mathbf{\hat{b}}_1\|, \|\mathbf{\hat{b}}_2\|, \cdots,
\|\mathbf{\hat{b}}_n\|\}$. In fact, using proof by contradiction, we can show that sorted GSO
minimizes $\max\{\|\mathbf{\hat{b}}_1\|, \|\mathbf{\hat{b}}_2\|, \cdots, \|\mathbf{\hat{b}}_n\|\}$.
This is in contrast (but also very similar in another sense) to V-BLAST which maximizes
$\min\{\|\mathbf{\hat{b}}_1\|, \|\mathbf{\hat{b}}_2\|, \cdots, \|\mathbf{\hat{b}}_n\|\}$
\cite{vblast:conf}.

%
%
%
%
%

\begin{algorithm}[t]
\caption{\quad Sorted GSO} \label{alg:SQR}
{\bf{Input}:} A basis $\mathbf{B} = [\mathbf{b}_1, ... \mathbf{b}_n]$\\
{\bf{Output}:} GSO for the sorted basis
\begin{algorithmic}[1]

\STATE let $\mathbf{\hat{B}} = \mathbf{B}$

\FOR{$i = 1, 2, ..., n$}

\STATE $k=\arg\min_{i \leq m \leq n}{\|\mathbf{\hat{b}}_m\|}$

\STATE exchange the $i$ and $k$-th columns of $\mathbf{\hat{B}}$

\FOR{$j = i+1, ..., n$}

\STATE compute the coefficient $\mu_{ij} =
\frac{\langle\mathbf{\hat{b}}_j,\mathbf{\hat{b}}_i\rangle}
  {\|\mathbf{\hat{b}}_i\|^2}$

\STATE update $\mathbf{\hat{b}}_j := \mathbf{\hat{b}}_j - \mu_{ij}\mathbf{\hat{b}}_i$

\STATE \%\% \textit{joint sorted GSO and size reduction} \%\%

\STATE \% $\mathbf{{b}}_j := \mathbf{{b}}_j - \lceil \mu_{ij}\rfloor \mathbf{{b}}_i$

\ENDFOR

\ENDFOR

\end{algorithmic}
\end{algorithm}


%
%
%
%
%
%
%



Following sorted GSO, it is natural to define the following notion of lattice reduction:
\begin{equation}\label{PDeep}
\begin{split}
  |\Re(\mu_{i,j})| &\leq 1/2, \quad |\Im(\mu_{i,j})|\ \leq 1/2, \quad \text{for} \quad 1 \leq j < i \leq n; \\
  \|\pi_i(\mathbf{b}_i)\| &=\min\{\|\pi_i(\mathbf{b}_i)\|,\|\pi_i(\mathbf{b}_{i+1})\|,...,
  \|\pi_i(\mathbf{b}_n)\|\},\\
  &\quad \text{for} \quad 1\leq i < n.
\end{split}
\end{equation}
In words, the basis is size-reduced and sorted in the sense of sorted GSO. Such a basis obviously
exists, as a KZ-reduced basis satisfies the above two conditions \cite{lagarias}. In fact, KZ
reduction searches for the shortest vector in the lattice with basis $[\pi_i(\mathbf{b}_i)$, ...,
$\pi_i(\mathbf{b}_n)]$.

The sorting condition is in fact stronger than the Lov\'asz condition. Since $\mathbf{\hat{b}}_i +
\mu_{i,i-1} \mathbf{\hat{b}}_{i-1}=\pi_{i-1}(\mathbf{b}_i)$ and
$\mathbf{\hat{b}}_{i-1}=\pi_{i-1}(\mathbf{b}_{i-1})$, the Lov\'asz condition with $\delta=1$ turns
out to be
\begin{equation}\label{}
  \|\pi_{i-1}(\mathbf{b}_{i-1})\| \leq \|\pi_{i-1}(\mathbf{b}_i)\|, \text{ for } 1<i\leq n,
\end{equation}
which can be rewritten as
\begin{equation}
  \|\pi_i(\mathbf{b}_i)\|=\min\{\|\pi_i(\mathbf{b}_i)\|,\|\pi_i(\mathbf{b}_{i+1})\|\}, \quad 1\leq i <n.
\end{equation}
Obviously, this is weaker than the sorting condition for all Gram-Schmidt vectors. Therefore, the
reduction notion defined above is stronger than LLL reduction even with $\delta=1$, but is weaker
than KZ reduction.

Meanwhile, when LLL-deep terminates, the following condition is satisfied:
\begin{equation}\label{}
\delta\|\mathbf{\hat{b}}_i\|^2 \leq \sum_{j=i}^k\mu_{k,j}^2\|\mathbf{\hat{b}}_j\|^2 \quad \text{for
} i<k\leq n.
\end{equation}
If $\delta = 1$, this is equivalent to
$\|\pi_i(\mathbf{b}_i)\|=\min\{\|\pi_i(\mathbf{b}_i)\|,\|\pi_i(\mathbf{b}_{i+1})\|,...,
  \|\pi_i(\mathbf{b}_n)\|\}$. Therefore, we have

\begin{prop}
The notion of lattice reduction defined in (\ref{PDeep}) is equivalent to LLL-deep with $\delta=1$
and unbounded window size.
\end{prop}

%
%
%
%
%

Although this notion is the same as LLL-deep, it offers an alternative implementation as shown in
Algorithm \ref{alg:NEW} which iterates between sorting and size reduction. Again, we refer to
sorting and size reduction as a super-iteration, since each sorting is equivalent to many swaps.
Obviously, the sorted GSO preceding the main loop is not mandatory; we show the algorithm in this
way for convenience of comparison with the DOLLAR detector \cite{Waters:Dollar} later on. Size
reduction does not change the GSO, but it shortens the vectors. Thus, after size reduction the
basis vector $\mathbf{b}_1$ may not be the shortest vector any more, and this may also happen to
other basis vectors. Then, the basis vectors are sorted again. The iteration will continue until
reduced basis is obtained in the end.

\begin{algorithm}[t]
\caption{\quad Parallel LLL-Deep} \label{alg:NEW}
{\bf{Input}:} A basis $\mathbf{B} = [\mathbf{b}_1, ... \mathbf{b}_n]$\\
{\bf{Output}:} The LLL-deep-reduced basis
\begin{algorithmic}[1]

\STATE sorted GSO of the basis

\WHILE {there is any update}

\STATE size reduction of the basis

\STATE sorted GSO of the basis

\ENDWHILE

\end{algorithmic}
\end{algorithm}

Obviously, Algorithm \ref{alg:NEW} finds a LLL-deep basis if it terminates. It is not difficult to
see that Algorithm \ref{alg:NEW} indeed terminates, by using an argument similar to that for
standard LLL with $\delta = 1$ \cite{Akhavi03, LenstraFlags}. The argument is that the order of the
vectors changes only when a shorter vector is inserted, but the number of vectors shorter than a
given length in a lattice is finite. Therefore, the iteration cannot continue forever. We
conjecture it also converges in $O(n\log n)$ super-iterations; in practice it seems to converge in
$O(n)$ super-iterations. Fig. \ref{fig:convergence} shows a typical outcome of numerical
experiments on the LLL potential function (\ref{potential}) against the number of super-iterations.
It is seen that parallel LLL-deep could decrease the potential more than LLL with $\delta = 1$, and
the most significant decrease occurs during the first few super-iterations.


\begin{figure}[t]

\centering\centerline{\epsfig{figure=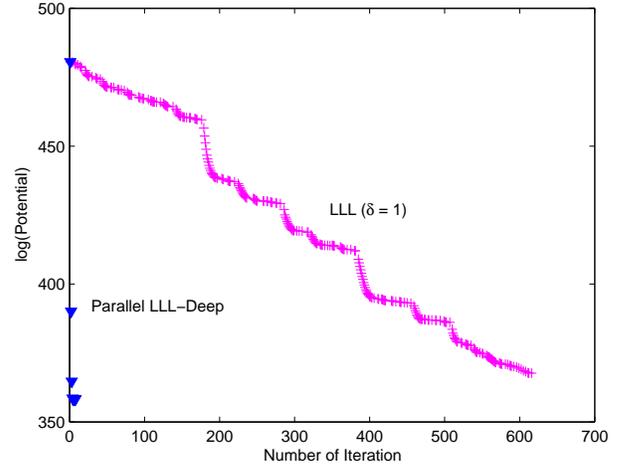,width=9cm}}

\caption{An example of the potential function against the number of iterations for standard LLL
with $\delta=1$ and parallel LLL-deep.}

\label{fig:convergence}
\end{figure}

As shown in Fig. \ref{fig:diagram}(a), the proposed parallel LLL-deep has the advantage of a
regular, modular structure, and further allows for pipeline implementation. Further, it is possible
to run sorted GSO and size reduction simultaneously in Algorithm \ref{alg:NEW}, as shown in Fig.
\ref{fig:diagram}(b). To do so, we just add Line 9 in sorted GSO (Algorithm \ref{alg:SQR}), which
will lead to the same reduced basis. It will cost approximately $n^3$ flops. Thus, the
computational complexity of each super-iteration is roughly $3n^3$, $50\%$ higher than that of
sorted GSO. Since sorted GSO and size reduction are computed simultaneously, the latency will be
reduced. Further, both sorted GSO and size reduction themselves can be parallelized \cite{Heckler}.
We can see that while the overall complexity might be $O(n^4\log n)$, the throughput and latency of
Fig. \ref{fig:diagram}(b) in a pipeline structure are similar to those of V-BLAST.

\begin{figure}[t]

\centering\centerline{\epsfig{figure=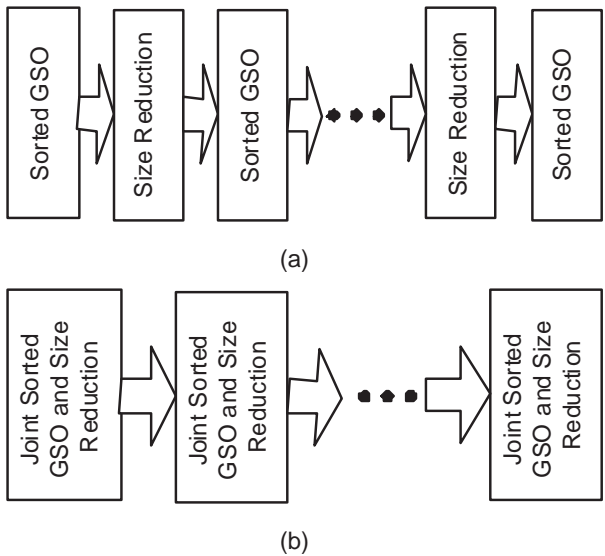,width=8cm}}

\caption{Fixed-complexity implementation of the LLL algorithm. (a) Separate sorted GSO and size
reduction; (b) joint sorted GSO and size reduction.}

\label{fig:diagram}
\end{figure}

\subsubsection{Using Sorted Cholesky Decomposition}

The complexity of the proposed parallel LLL-deep in Fig. \ref{fig:diagram} mostly comes from
repeated GSO. To reduce the complexity, we can replace it by sorted Cholesky decomposition
\cite{Ling:JSTSP09}. This will yield the same reduced basis, but has lower computational
complexity. The complexity of sorted Cholesky decomposition is approximately $n^3/3$ (the initial
multiplication $\mathbf{A}=\mathbf{B}^H\mathbf{B}$ costs approximately $n^3$ due to symmetry),
while that of sorted QR decomposition is approximately $2n^3$.

Since the sorted Cholesky decomposition was given for the dual basis in \cite{Ling:JSTSP09}, for
the sake of clarity we redescribe it in Algorithm~\ref{alg:SCHOL}. Here, $c_{m,n}$ is the
$(m,n)$-th entry of $\mathbf{C}$, while $\overline{c_{m,n}}$ denotes its complex conjugate. For
convenience, we also use MATLAB notation $c_{i:j,k}$ to denote a vector containing those elements
of $\mathbf{C}$. When Algorithm \ref{alg:SCHOL} terminates, the lower triangular part of
$\mathbf{C}$ is the Hermitian transpose of the R factor of the QR decomposition for the basis
$\mathbf{B}$. Similarly, one can run sorted Cholesky decomposition and size reduction
simultaneously.


\begin{algorithm}[t]
\caption{\quad Sorted Cholesky Decomposition} \label{alg:SCHOL}
{\bf{Input}:} Gram matrix $\mathbf{A}=\mathbf{B} ^H\mathbf{B}$\\
{\bf{Output}:} R factor for the sorted basis
\begin{algorithmic}[1]

\STATE let $\mathbf{{C}} = \mathbf{A}$

\FOR{$i = 1, 2, ..., n$}

\STATE $k=\arg\min_{i \leq m \leq n}{a_{m,m}}$


\STATE exchange the $i$ and $k$-th columns and rows of $\mathbf{{C}}$

\STATE $c_{i,i} := \sqrt{c_{i,i}}$

\STATE $c_{i+1:n,i} := \frac{c_{i+1:n,i}}{{c_{i,i}}}$

\FOR{$j = i+1, ..., n$}

\STATE $c_{j:n,j} := c_{j:n,j} - c_{j:n,i}\overline{c_{j,i}}$

\ENDFOR

\ENDFOR

\end{algorithmic}
\end{algorithm}

\subsubsection{Relation to V-BLAST and DOLLAR Detector}

The V-BLAST ordering is a well known technique to improve performance of communication by
pre-sorting the columns of a matrix \cite{vblast:conf}. It maximizes the length of the shortest
Gram-Schmidt vector of $\mathbf{B}$ among all $n!$ possible orders. V-BLAST ordering starts from
the last column of $\mathbf{B}$; it successively chooses the $i$-th Gram-Schmidt vector with the
maximum length for $i = n, n-1, ..., 1$. Several $O(n^3)$ algorithms have been proposed. One of
them is obtained by applying sorted GSO to the dual lattice \cite{Ling:JSTSP09}. This results in
significant computational savings because only a single GSO process is needed.

Now it is clear that the first iteration of parallel LLL-deep is very similar to the DOLLAR
detector in \cite{Waters:Dollar}, which is comprised of a sorted QR decomposition, a size
reduction, and a V-BLAST sorting. Since V-BLAST sorting is very close to sorted QR decomposition,
replacing V-BLAST ordering with sorted QR decomposition does not make much difference. In view of
this, the DOLLAR detector can be seen as the first-order approximation of parallel LLL-deep, which
explains its good performance. It can be seen from Fig. \ref{fig:convergence} that the first
iteration appears to decrease the potential more than any other iterations. Of course, using just
one iteration also limits the performance of the DOLLAR detector.


%

\subsubsection{Relation to Standard LLL and Some Variants}

In fact, the inventors of the LLL algorithm already suggested to successively make
$\|\mathbf{\hat{b}}_1\|, \|\mathbf{\hat{b}}_2\|, \cdots, \|\mathbf{\hat{b}}_n\|$ as small as
possible \cite{BK:Lovasz86}, since they observed that shorter vectors among $\mathbf{\hat{b}}_1,
\mathbf{\hat{b}}_2, \cdots, \mathbf{\hat{b}}_n$ typically appear at the end of the sequence. This
idea is similar to sorted QR decomposition, but in the LLL algorithm it is implemented
incrementally by means of swapping (i.e., in a bubble-sort fashion).



Joint sorting and reduction \cite{Gan:JSR} is also similar to LLL-deep. It is well known that
ordering can be used as a preprocessing step to speed up lattice reduction. A more natural approach
is joint sorting and reduction \cite{Gan:JSR} that uses modified GSO and when a new vector is
picked it picks the one with the minimum norm (projected to the orthogonal complement of the basis
vectors already reduced). That is, it runs sorted GSO only once.




%

Recently, Nguyen and Stehl\'e proposed a greedy algorithm for low-dimensional lattice reduction
\cite{NgSt04}. It computes a Minkowski-reduced basis up to dimension 4. Their algorithm is
recursive; in each recursion, the basis vectors are ordered by increasing lengths. Obviously,
following their idea, we can define another notion of lattice reduction where the basis vectors are
sorted in increasing lengths and also size-reduced. The implementation of this algorithm resembles
that of parallel LLL-deep. That is, such a basis can be obtained by alternatively sorting in
lengths and size reduction. Using the same argument as that of LLL-deep, one can show that this
algorithm will also terminate after a finite number of super-iterations. However, it seems
difficult to prove any bounds for this algorithm.

\subsubsection{Reducing the Complexity of Sequential LLL-Deep}



While the primary goal is to allow parallel pipeline hardware implementation, the proposed LLL-deep
algorithm also has a computational advantage over the conventional LLL-deep algorithm even in a
sequential computer for the first few iterations. We observed that running parallel LLL-deep
crudely will not improve the speed. While parallel LLL-deep is quite effective at the early stage,
keeping running it becomes wasteful at the late stage as the quality of the basis has improved a
lot. In fact, updating occurs less frequently at the late stage; thus the standard serial version
of LLL-deep will be faster. As a result, a hybrid strategy using parallel LLL-deep at the early
stage and then switching to the serial version at the late stage will be more efficient. Parallel
LLL-deep can be viewed as a preprocessing stage for such a hybrid strategy. One can run some
numerical experiments to determine when is the best time to switch from parallel to sequential
LLL-deep.

In Fig. \ref{fig:MIMOtime}, we show the average running time for LLL-deep for the complex MIMO
lattice, on a notebook computer with Pentium Dual CPU working at 1.8 GHz. Sorted GSO is used. It is
seen that the hybrid strategy can improve the speed by a factor up to 3 for $n \leq 50$.


\begin{figure}[t]

\centering\centerline{\epsfig{figure=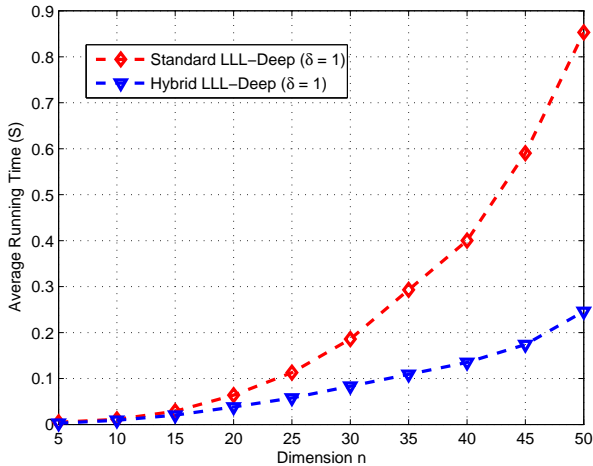,width=9cm}}

\caption{Average running time for standard and hybrid LLL-deep for the complex MIMO lattice. The
number of iterations in parallel LLL-deep is set to 2.}


\label{fig:MIMOtime}
\end{figure}

\subsection{Bit Error Rate (BER) Performance}


We evaluate the impact of a finite number of super-iterations on the performance of parallel LLL
and LLL-deep. In the following simulations of the BER performance, we use MMSE-based lattice
decoding and set $\delta=1$ in the complex LLL algorithm for the best performance.


Fig. \ref{fig:ELLL8x} shows the performance of parallel effective LLL for different numbers of
super-iterations for an $8\times 8$ MIMO system with 64-QAM modulation and SIC detection. The
performance of ML detection is also shown as a benchmark of comparison. It is seen that increasing
the number of iterations improves the BER performance; in particular, with 8 iterations, parallel
effective LLL almost achieves the same performance as standard LLL. On the other hand, the
returning SNR gain after the first few iterations is diminishing as the number of iterations
increases.

\begin{figure}[t]

\centering\centerline{\epsfig{figure=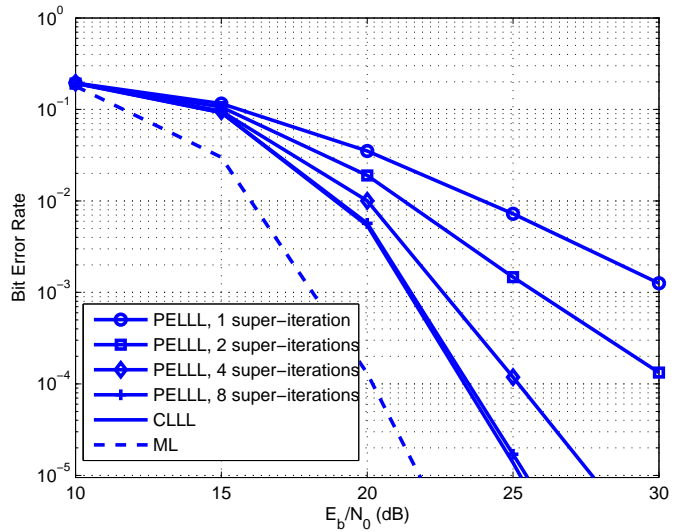,width=10cm}}

\caption{Performance of parallel effective LLL (PELLL) for a $8\times 8$ MIMO system with 64-QAM
and SIC detection.}

\vspace{-0.5cm}

\label{fig:ELLL8x}
\end{figure}

Fig. \ref{fig:DEEP8x} shows the performance of parallel LLL-deep for an $8\times 8$ MIMO system
with 64-QAM modulation and SIC detection. A similar trend is observed.  Note that parallel LLL-deep
with only one super-iteration, which essentially corresponds to the DOLLAR detector in
\cite{Waters:Dollar}, does not achieve full diversity, despite its good performance. Compared to
Fig. \ref{fig:ELLL8x}, we can see that the performance is very similar in the end, but parallel
LLL-deep seems to perform better in the first few super-iterations.

\begin{figure}[t]

\centering\centerline{\epsfig{figure=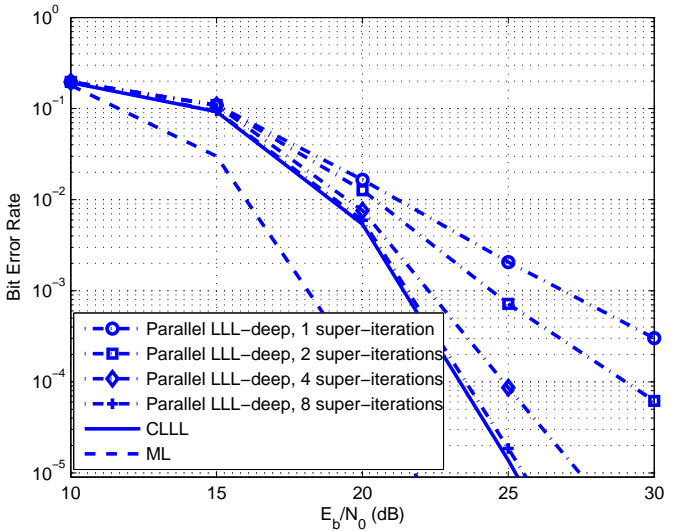,width=10cm}}

\caption{Performance of parallel LLL-deep for an $8\times 8$ MIMO system with 64-QAM and SIC
detection.}

\vspace{-0.5cm}

\label{fig:DEEP8x}
\end{figure}


\section{Concluding Remarks}

We have derived the $O(n^4 \log n)$ average complexity of the LLL algorithm in MIMO communication.
We also proposed the use of effective LLL that enjoys $O(n^3 \log n)$ theoretic average complexity.
Although in practice effective LLL does not significantly reduce the complexity, the $O(n^3 \log
n)$ theoretic bound improves our understanding of the complexity of LLL. To address the issue of
variable complexity, we have proposed two parallel versions of the LLL algorithm that allow
truncation after some super-iterations. Such truncation led to fixed-complexity approximation of
the LLL algorithm. The first such algorithm was based on effective LLL, and we argued that it is
sufficient to run $O(n \log n)$ super-iterations. The second was a parallel version of LLL-deep,
which overcomes the limitation of the DOLLAR detector. We also showed that V-BLAST is a relative of
LLL.

Finally, we point out some open questions.

A precise bound on the complexity of LLL remains to be found. The $O(n^4 \log n)$ bound seems to be
loose in MIMO communication.

Using effective LLL reduction may raise the concern of numerical stability, as the GS coefficients
will grow. Although with the accuracy present in most floating-point implementations, this does not
seem to cause a problem for the practical ranges of $n$ in MIMO communications, a rigorous study on
this issue (e.g., following \cite{Nguyen2}) is a topic of future research.

Although for parallel effective LLL, we proved that the first basis vector is short after $O(n \log
n)$ super-iterations, it remains to show in theory whether full diversity can be achieved or not.
It is also an open question whether similar results exist for parallel LLL-deep.



%


%
%
\appendices

\section{Relations Between Real and Complex LLL} \label{AppendixI}

\subsection{Reducedness}

It is common in literature to convert the complex basis matrix ${\bf B}$ into a real matrix and
then apply the standard real LLL algorithm. We shall analyze the relationship between this approach
and complex LLL. There are many ways to convert ${\bf B}$ into a real matrix. One of them is to
convert each element of ${\bf B}$ locally, i.e.,
\begin{equation} \label{real-equiv2}
  {\bf B_{\text{R}}} =
  \left[%
    \begin{array}{ccc}
      \Re(b_{1,1}) & -\Im(b_{1,1}) & \cdots\\
      \Im(b_{1,1}) & \Re(b_{1,1}) & \cdots\\
      \cdots& \cdots& \ddots\\
    \end{array}%
\right].
\end{equation}
For convenience, we write this conversion as \[\mathbf{B}_{\text{R}}=\mathcal{F}(\mathbf{B}).\]
Obviously, ${\bf B_{\text{R}}}$ is $2n$-dimensional if ${\bf B}$ is $n$-dimensional.

Another conversion is
\begin{equation} \label{real-equiv1}
  {\bf B}_{\text{R}}' = \left[%
    \begin{array}{cc}
      \Re({\bf B}) & -\Im({\bf B}) \\
      \Im({\bf B}) & \Re({\bf B}) \\
    \end{array}%
\right].
\end{equation}
Correspondingly, we write this as \[\mathbf{B}_{\text{R}}'=\mathcal{F}'(\mathbf{B}).\]

Note that when the real-valued LLL algorithm is applied to $\mathbf{B}_{\text{R}}$ or
$\mathbf{B}_{\text{R}}'$, the structures as above are generally not preserved.



Now suppose the complex matrix ${\bf B}$ has been complex LLL-reduced. We then expand the reduced
${\bf B}$ as in (\ref{real-equiv2}) or (\ref{real-equiv1}). What can be said about
$\mathbf{B}_{\text{R}}$ and $\mathbf{B}_{\text{R}}'$?

\begin{lem}\label{lemma1}
Let $\mathbf{B}=\mathbf{\hat{B}}\bm{\mu}^T$ and
$\mathbf{B}_{\text{R}}=\mathbf{\hat{B}}_{\text{R}}(\bm{\mu}_{\text{R}})^T$ be the Gram-Schmidt
orthogonalization of the complex matrix $\mathbf{B}$ and its real counterpart
$\mathbf{B}_{\text{R}}$, respectively. Then we have $\mathbf{\hat{B}}_{\text{R}}=
\mathcal{F}(\mathbf{\hat{B}})$ and $\bm{\mu}_{\text{R}} = \mathcal{F}(\bm{\mu})$.
\end{lem}

The proof is omitted. Lemma \ref{lemma1} says that the structure in (\ref{real-equiv2}) is
preserved under Gram-Schmidt orthogonalization. Using this we now prove the following result.

\begin{prop}\label{Prop1}
If $\mathbf{B}$ is reduced in the sense of complex LLL with parameter $\delta$ ($1/2 < \delta \leq
1$), then $\mathbf{B}_{\text{R}}$ is reduced in the sense of real LLL with parameter $\delta-1/4$.
\end{prop}

\begin{proof}
First, note that
$\|\mathbf{\hat{b}}_{\text{R},2i-1}\|=\|\mathbf{\hat{b}}_{\text{R},2i}\|=\|\mathbf{\hat{b}}_{i}\|$,
and $\bm{\mu}_{\text{R}}$ looks like
\begin{equation} \label{}
  {\bm \mu}_{\text{R}} =
  \left[%
    \begin{array}{ccccc}
      1 & 0 & 0& \cdots & \cdots\\
      0 & 1 & 0& \cdots & \cdots\\
      \Re(\mu_{2,1}) & -\Im(\mu_{2,1}) & 1 & \cdots& \cdots\\
      \Im(\mu_{2,1}) & \Re(\mu_{2,1}) & 0 & 1& \cdots\\
      \cdots  & \cdots & \cdots & \cdots& \ddots\\
    \end{array}%
\right].
\end{equation}
Obviously the Lov\'asz condition is satisfied by the $(2i-1)$ and $(2i)$-th column vectors of
$\mathbf{B}_{\text{R}}$ because
$\|\mathbf{\hat{b}}_{\text{R},2i}\|^2=\|\mathbf{\hat{b}}_{\text{R},2i-1}\|^2 \geq
(\delta-|\mu_{\text{R},2i,2i-1}|^2) \|\mathbf{\hat{b}}_{\text{R},2i-1}\|^2$. Let us examine $(2i)$
and $(2i+1)$-th column vectors. From (\ref{lll-cond-4}) we have
\begin{equation} \label{}
\begin{split}
    \|\mathbf{\hat{b}}_{\text{R},2i+1}\|^2 &\geq (\delta - |\Re(\mu_{i+1,i})|^2-
    |\Im(\mu_{i+1,i})|^2)\|\mathbf{\hat{b}}_{\text{R},2i}\|^2\\
    &\geq (\delta - 1/4 -
    |\Im(\mu_{i+1,i})|^2)\|\mathbf{\hat{b}}_{\text{R},2i}\|^2\\
    &= (\delta - 1/4 -
    |\mu_{\text{R},2i+1,2i}|^2)\|\mathbf{\hat{b}}_{\text{R},2i}\|^2.
\end{split}
\end{equation}
This proves the Proposition.
\end{proof}

\begin{rem}
It appears that little can be said about $\mathbf{B}_{\text{R}}'$.
\end{rem}


By Proposition \ref{Prop1}, even if a 2-dimensional complex lattice $\mathbf{B}$ is Gauss-reduced
(which is arguably the strongest reduction), its real equivalent $\mathbf{B}_{\text{R}}$ is not
necessarily LLL-reduced with parameter $\delta = 1$. It is easy to construct such a lattice
\[  \left[%
    \begin{array}{cc}
      1 & 0 \\
      1/2+j/2 & \sqrt{2}/2 \\
    \end{array}%
\right].\]


\subsection{Approximation factors}







Let us compare the approximation factors for the first vectors $\mathbf{b}_1$ and
$\mathbf{b}_{\text{R},1}$ in the reduced bases. Since $\det \mathbf{B} = \text{det} ^2
\mathbf{B}_{\text{R}}$ and the shortest nonzero vectors have the same length in the real and
complex bases, we only need to compare $\alpha^{n-1}$ with $\beta^{2n-1}$. Recall that
$\alpha=1/(\delta-1/2)$ and $\beta=1/(\delta-1/4)$. It is easy to check $\alpha \geq \beta^2$,
because
\begin{equation}
    (\delta-1/4)^2\geq\delta-1/2
\end{equation}
where the equality holds if and only if $\delta = 3/4$. Therefore, asymptotically, complex LLL has
a larger approximation factor unless $\delta = 3/4$. In fact, a minor difference between the error
performances of real and complex LLL can be observed when $n$ becomes large (e.g., $n \geq 10$),
although they are indistinguishable at small dimensions \cite{ComplexLLL}.

\section*{Acknowledgment}
The authors would like to thank D. Stehl\'e and J. Jald\'en for helpful discussions.


\footnotesize
\bibliographystyle{IEEEtran}
\bibliography{IEEEabrv,lingbib}

%





\end{document}